\documentclass[11pt,oneside]{article}
\usepackage[centertags]{amsmath}
\usepackage{amssymb}
\usepackage{amsthm}
\usepackage{hyperref}
\newtheorem{theorem}{Theorem}[section]

\newtheorem{proposition}{Proposition}[section]

\newcommand{\punt}{\boldsymbol{.}}
\begin{document}
\title{On the computation of classical, boolean and free cumulants}
\author{E. Di Nardo \footnote{Dipartimento di Matematica e Informatica,
Universit\`a degli Studi della Basilicata, Viale dell'Ateneo
Lucano 10, 85100 Potenza, Italia, elvira.dinardo@unibas.it}, I.
Oliva \footnote {imma.oliva@hotmail.it}}
\date{\today}
\maketitle
\begin{abstract}
This paper introduces a simple and computationally efficient
algorithm for conversion formulae between moments and
cumulants. The algorithm provides just one formula for classical,
boolean and free cumulants. This is realized by using a suitable
polynomial representation of Abel polynomials. The algorithm
relies on the classical umbral calculus, a symbolic language
introduced by Rota and Taylor in \cite{SIAM}, that is particularly
suited to be implemented by using software for symbolic
computations. Here we give a {\tt MAPLE} procedure. Comparisons
with existing procedures, especially for conversions between
moments and free cumulants, as well as examples of applications to
some well-known distributions (classical and free) end the paper.
\end{abstract}
\textsf{\textbf{keywords}:  umbral calculus, classical cumulant,
boolean cumulant, free cumulant, Abel polynomial}\\\\
\textsf{\textsf{AMS subject classification}: 65C60, 05A40, 46L53}\\
%
\section{Introduction}

Among number sequences connected to random variables (r.v.'s),
cumulants play a central role, characterizing many r.v.'s found in
the practice. Moreover, due to their properties of additivity and
invariance under translation, cumulants are not necessarily
connected with moments of r.v.'s so that they can be analyzed by
using an algebraic point of view. For these reasons, cumulants
have been introduced not only in the context of classical
probability theory, but also in the boolean \cite{Speicher} and
the free context \cite{Voiculescu}. Since cumulants linearize
convolutions of measures - no matter what framework they are
referred to - this property allows us quick access to test whether
a given probability measure is a convolution. The linearity of
classical convolutions corresponds to independent r.v.'s
\cite{Feller} and the linearity of free convolutions allows us to
recognize free r.v.'s \cite{Voiculescu}. The boolean convolution
was constructed \cite{Speicher} starting exactly from the notion
of partial cumulants, this last extensively used in the context of
stochastic differential equations.

In this paper, we propose the classical umbral calculus of Rota
and Taylor \cite{SIAM} as the most natural syntax for conversions between cumulants
(classical, boolean and free) and moments, without any sophisticated
programming.

The basic devices of the umbral calculus are essentially two. The
first one is to represent a unital sequence of numbers by a symbol
$\alpha,$ called an umbra, that is, to represent the sequence $1,
a_1, a_2, \ldots$ by means of the sequence $1, \alpha, \alpha^2,
\ldots$ of powers of $\alpha$ via a linear operator $E,$ resembling the
expectation operator of r.v.'s. This setting is quite natural in free
probability \cite{Voiculescu}. The second device consists in representing
a sequence $1, a_1, a_2,\ldots,$ by means of distinct umbrae,
as it happens in probability theory too for independent
and identically distributed (i.i.d.) r.v.'s. Thanks
to these devices above all, the umbral calculus was already
employed successfully as a syntax for symbolic computations
in statistics \cite{Dinardo1, Dinardo2} by using {\tt Maple}.
Multivariate extensions have been given in \cite{Bernoulli} for
cumulant estimators as an alternative to tensor methods \cite{McCullagh}.

The algorithm here proposed relies on umbral parametrizations of
cumulants (classical, boolean and free) in terms of moments and
viceversa, carried out in \cite{free}. In the free context, such a
parametrization involves umbral Abel polynomials. Here we show
that it is sufficient to change a parameter in the formula in
order to get same conversions in classical and boolean theory. So
we trace back all the parametrizations to one umbral polynomial.
We give a suitable expansion of such a polynomial in order to get
a very simple algorithm for the whole matter.

The paper is organized as follows. Section 2 is provided for
readers unaware of the classical umbral calculus. We resume
terminology, notation and some basic definitions. Section 3
recalls the umbral theory of cumulants both classical, boolean and
free, as well as their corresponding parametrizations. In Section
4, we show that all the parametrizations can be recovered through
only one umbral polynomial for which a fruitful expansion is
provided. The {\tt MAPLE} algorithm for this expansion is also
given. Comparisons, with procedures available in the literature
(see for example \cite{Bryc}), confirm the competitiveness of the
umbral algorithm. The paper ends with some examples of how to use
umbral calculus in order to compute classical, boolean and free
cumulants for some classical probability laws.

\section{The classical umbral calculus}
In the following, we recall terminology, notation and some basic
definitions of the classical umbral calculus, as introduced
by Rota and Taylor in \cite{SIAM} and  subsequently developed by
Di Nardo and Senato in \cite{Dinardo} and \cite{Dinardoeurop}. We also recall those
results useful in developing the umbral algorithm, given in
Section 4. We skip any proof: the reader interested in-depth
analysis is referred to \cite{Dinardo} and \cite{Dinardoeurop}.

Classical umbral calculus is a syntax consisting of the following
data: a set $A=\{\alpha,\beta, \ldots \},$ called the
{\it alphabet}, whose elements are named {\it umbrae}; a
commutative integral domain $R$ whose quotient field is of
characteristic zero\footnote{In many applications $R$ is the field of real
or complex numbers.}; a linear functional $E: R[A] \rightarrow R,$
called {\it evaluation}, such that $E[1]=1$ and
$$E[\alpha^i \beta^j \cdots \gamma^k] = E[\alpha^i]E[\beta^j] \cdots E[\gamma^k]$$
for any set of distinct umbrae in $A$ and for $i,j,\ldots,k$
nonnegative integers (the so-called {\it uncorrelation property});
an element $\epsilon \in A,$ called {\it augmentation} \cite{Roman},
such that $E[\epsilon^n] = \delta_{0,n},$ for any nonnegative integer $n,$ where
$\delta_{i,j} = 1$ if $i=j,$ otherwise being zero; an element $u \in A,$ called {\it unity} umbra,
such that $E[u^n]=1,$ for any nonnegative integer $n.$

A sequence $a_0=1,a_1,a_2, \ldots$ in $R$ is umbrally represented
by an umbra $\alpha$ when
$$E[\alpha^i]=a_i, \quad \hbox{for} \,\, i=0,1,2,\ldots.$$
The elements $a_i$ are called {\it moments} of the umbra $\alpha.$
This name recalls the device, familiar to statisticians, when
$a_i$ represents the $i$-th moment of a r.v. $X$.
Similarly, the {\it factorial moments} of an umbra $\alpha$ are
the elements
$$a_{(0)}=1, \quad a_{(n)} = E[(\alpha)_n] \quad \hbox{\rm for nonnegative integers} \,n$$
where $(\alpha)_n=\alpha(\alpha-1)\cdots(\alpha-n+1)$ is the lower
factorial.
The following umbrae play a special role in the umbral calculus. \\
{\bf Singleton umbra.} The singleton umbra $\chi$ is the
umbra such that $E[\chi^1] = 1$ and $E[\chi^n] = 0$ for
$n=2,3,\ldots.$ Its factorial moments are $x_{(n)} = (-1)^{n-1}
(n-1)!$. \\
{\bf Bell umbra.} The Bell umbra $\beta$ is the umbra whose
factorial moments are all equal to $1$, i.e. $E[(\beta)_n] = 1$
for $n=0,1,2,\ldots.$ Its moments are the Bell numbers, that is
the $n$-th coefficient in the Taylor series expansion of the
function $\exp(e^t-1).$

An umbral polynomial is a polynomial $p \in R[A].$ The support of
$p$ is the set of all umbrae occurring in $p.$ If $p$ and $q$ are
two umbral polynomials, then  $p$ and $q$ are {\it uncorrelated} if and only if
their supports are disjoint; $p$ and $q$ are {\it
umbrally equivalent} if and only if
$$E[p] =E[q], \quad \hbox{\rm in symbols} \quad p \simeq q.$$
It is possible that two distinct umbrae represent the same
sequence of moments, in such case they are called {\it similar
umbrae}. More formally two umbrae $\alpha$ and $\gamma$ are {\it
similar} when $\alpha^n$ is umbrally equivalent to $\gamma^n,$ for
all $n=0,1,2,\ldots$ in symbols
$$\alpha \equiv \gamma \Leftrightarrow \alpha^n \simeq \gamma^n \quad
n=0,1,2,\ldots.$$ Given a sequence $1, a_1, a_2,
\ldots$ in $R$ there are infinitely many distinct, and thus
similar umbrae, representing the sequence.

Thanks to the notion of similar umbrae, the alphabet $A$ can be extended
by inserting the so-called {\it auxiliary} umbrae,
resulting from operations among similar umbrae. This leads to
construct a saturated umbral calculus, in which auxiliary umbrae
are handled as elements of the alphabet.

In the following, we focus the attention on some special auxiliary umbrae. We assume
$\{\alpha^{\prime},\alpha^ {\prime \prime},\ldots,
\alpha^{\prime\prime\prime}\}$ is a set of $n$ uncorrelated umbrae
similar to an umbra $\alpha$. \\
{\bf Dot power.} The symbol $\alpha^ {\, \punt n}$ is an auxiliary umbra
denoting the product $\alpha^{\prime} \, \alpha^{\prime \prime}
\cdots \alpha^{\prime\prime\prime}.$ Moments of $\alpha^{\, \punt n}$
can be easily recovered from its definition. Indeed, if the umbra
$\alpha$ represents the sequence $1, a_1, a_2, \ldots,$ then
$E[(\alpha^{\, \punt n})^k] = a_k^n$ for all nonnegative integers
$k$ and $n$. \\
{\bf Dot product.} The symbol $n \punt \alpha$ denotes an
auxiliary umbra similar to the sum $\alpha^{\prime}+\alpha^{\prime
\prime}+ \cdots + \alpha^{\prime\prime\prime}.$ So $n \punt
\alpha$ is the umbral counterpart of a sum of i.i.d. r.v.'s.
Moments of $n \punt \alpha$ can be expressed using the notions of
integer partition\footnote{Recall that a partition of an integer
$i$ is a sequence $\lambda = (\lambda_1, \lambda_2, \ldots,
\lambda_t),$ where $\lambda_j$ are weakly decreasing positive
integers such that $\sum_{j=1}^{t} \lambda_j = i.$ The integers
$\lambda_j$ are named {\it parts} of $\lambda.$ The {\it lenght}
of $\lambda$ is the number of its parts and will be indicated by
$\nu_{\lambda}.$  A different notation is $\lambda = (1^{r_1},
2^{r_2}, \ldots),$ where $r_j$ is the number of parts of $\lambda$
equal to $j$ and $r_1 + r_2 + \cdots = \nu_{\lambda}.$ Note that
$r_j$ is said to be the multiplicity of $j$. We use the classical
notation $\lambda \vdash i$ to denote \lq\lq $\lambda$ is a
partition of $i$\rq\rq.} and dot-power. By using an umbral version
of the well-known multinomial expansion theorem, we have
\begin{equation}
(n \punt \alpha)^i \simeq \sum_{\lambda \vdash i}
(n)_{\nu_{\lambda}} d_ {\lambda} \alpha_{\lambda}, \label{(eq:11)}
\end{equation}
where the sum is over all partitions $\lambda = (1^{r_1}, 2^{r_2},
\ldots)$ of the integer $i,$ $(n)_{\nu_{\lambda}}=0$ for
$\nu_{\lambda} > n,$
\begin{equation}
d_{\lambda} = \frac{i!}{r_1! r_2! \cdots} \,
\frac{1}{(1!)^{r_1}(2!)^{r_2} \cdots}  \quad \hbox{and} \quad
\alpha_{\lambda} = [\alpha^{\prime}]^{\punt \, r_1}
[(\alpha^{\prime \prime})^2]^{\punt \, r_2} \cdots.
\label{(eq:12)}
\end{equation}

A feature of the classical umbral calculus is the construction of
new auxiliary umbrae by suitable symbolic substitutions. For
example, in $n \punt \alpha$ replace the integer $n$ by an umbra
$\gamma$. From (\ref{(eq:11)}), the new auxiliary umbra $\gamma
\punt \alpha$ is such that
\begin{equation}
(\gamma \punt \alpha)^i \simeq \sum_{\lambda \vdash i}
(\gamma)_{\nu_{\lambda}} d_{\lambda} \alpha_{\lambda},
\label{(eq:13)}
\end{equation}
and it is the umbral counterpart of a so-called random sum. By using (\ref{(eq:13)}), we have
$\beta \punt \chi \equiv \chi \punt \beta \equiv u$ as well as $(\alpha \punt \chi)^i \simeq (\alpha)_i.$
In the next section, we will see that $\chi \punt \alpha$ has a different meaning. Moreover we have
\begin{equation}
(\gamma \punt \beta \punt \alpha)^i \simeq \sum_{\lambda \vdash i}
\gamma^{\nu_ {\lambda}} d_{\lambda} \alpha_{\lambda}.
\label{(eq:14)}
\end{equation}
The umbra $\gamma \punt \beta \punt \alpha$ is called {\it composition
umbra} of $\alpha$ and $\gamma$ and it is such that $(\gamma \punt
\beta) \punt \alpha \equiv \gamma \punt (\beta \punt \alpha).$ The
{\it compositional inverse} $\alpha^{<-1>}$ of an umbra $\alpha$ is such
that $\alpha^{<-1>} \punt \beta \punt \alpha \equiv \chi \equiv
\alpha \punt \beta \punt \alpha^{<-1>}$.
\section{Classical, boolean and free cumulants}
In this section, we recall some results given in
\cite{free} about connections between moments and cumulants
(classical, boolean and free).

\smallskip \noindent
{\bf $\alpha$-cumulant umbra.} The umbra $\chi \punt \alpha,$
where $\chi$ is the singleton umbra, is called $\alpha$-cumulant
umbra. Replacing $\gamma$ by $\chi$ and by virtue of
(\ref{(eq:13)}), we have
\begin{equation}
(\chi \punt \alpha)^i \simeq \sum_{\lambda \vdash i}
x_{\nu_{\lambda}} d_ {\lambda} \alpha_{\lambda} \simeq
\sum_{\lambda \vdash i} (-1)^{\nu_\lambda - 1} (\nu_\lambda - 1)!
\, d_ {\lambda} \, \alpha_{\lambda}. \label{(eq:15)}
\end{equation}
Since the second equivalence in (\ref{(eq:15)}) recalls the
well-known expression of cumulants in terms of moments of a r.v.,
it is natural to refer to moments of $\chi \punt \alpha$ as
cumulants of $\alpha$. The $\alpha$-cumulant umbra, usually
denoted by $\kappa_{\scriptscriptstyle \alpha},$ is deeply studied
in \cite{Dinardoeurop}. In particular, if
$\kappa_{\scriptscriptstyle \alpha}$ is the $\alpha$-cumulant
umbra, then $\alpha \equiv \beta \punt \,
\kappa_{\scriptscriptstyle \alpha}.$ Moreover, by recalling
equivalence (\ref{(eq:13)}) and $E[(\beta)_i] = 1$ for all
nonnegative integers $i,$ we have
\begin{equation}
\alpha^i \simeq (\beta \punt \, \kappa_{\scriptscriptstyle \alpha})^i \simeq
\sum_{\lambda \vdash i} (\beta)_{\nu_{\lambda}} d_{\lambda}
(\kappa_{\scriptscriptstyle \alpha})_{\lambda} \simeq
\sum_{\lambda \vdash i} \, d_{\lambda} \,
(\kappa_{\scriptscriptstyle \alpha})_{\lambda}. \label{(eq:16)}
\end{equation}
\begin{theorem}[Parametrizations] Let $\kappa_{\scriptscriptstyle \alpha}$ be
the $\alpha$-cumulant umbra. For $i=1,2,\ldots$ we have
\begin{equation}
\alpha^i \simeq  \kappa_{\scriptscriptstyle \alpha}
(\kappa_{\scriptscriptstyle \alpha} + \beta \punt
\kappa_{\scriptscriptstyle \alpha})^{i-1} \quad
\kappa_{\scriptscriptstyle \alpha}^i \simeq \alpha( \alpha - 1
\punt \alpha)^{i- 1}. \label{par1}
\end{equation}
\end{theorem}
\noindent The symbol $-1 \punt \alpha$ denotes the inverse of an umbra $\alpha$, that is
the umbra such that $\alpha + -1 \punt \alpha \equiv \varepsilon.$

\smallskip \noindent
{\bf $\alpha$-boolean cumulant umbra.} Let $M(t)$ be the ordinary
generating function (g.f.) of a r.v. $X$, that is $M(t)= 1 +
\sum_{i \geq 1}  a_i t^i$ where $a_i = E[X^i]$. We have
$$M(t) = \frac{1}{1-H(t)}, \quad \hbox{where} \quad H(t)= \sum_{i \geq 1} h_i
t^i,$$ and $h_i$ are called boolean cumulants of $X$. The umbral
theory of boolean cumulants has been introduced in \cite{free}. In
particular, the $\alpha$-boolean cumulant umbra $\eta_
{\scriptscriptstyle\alpha}$ is such that $E[\eta_
{\scriptscriptstyle\alpha}^i]=h_i$. This umbra corresponds to the
composition of the umbra $\bar{\alpha},$ having moments
$\bar{\alpha}^i \simeq i! \alpha_i,$ and the compositional inverse
of $\bar{u},$ having moments $n!,$ in symbols
$\bar{\eta}_{\scriptscriptstyle\alpha} \equiv \bar{u}^{<-1>} \punt
\beta \punt \bar{\alpha}.$ By this last equivalence, we get
$$\bar{\eta}_{\scriptscriptstyle\alpha}^i \simeq \sum_{\lambda \vdash i} (-1)^
{\nu_\lambda-1} \nu_\lambda! \, d_{\lambda} \,
\bar{\alpha}_{\lambda}.$$ The previous equivalence allows us to
express boolean cumulants in terms of moments. In order to get the
inverted expressions, we need of the following equivalence
$\bar{\alpha} \equiv \bar{u} \punt \beta \punt
\bar{\eta}_{\scriptscriptstyle\alpha},$ that has been proved in
the Boolean Inversion Theorem (cfr. \cite{free}). Again we have
$$\bar{\alpha}^i \simeq \sum_{\lambda \vdash i}  \nu_\lambda! \, d_{\lambda} \,
(\bar{\eta}{\scriptscriptstyle\alpha})_{\lambda}.$$ Observe the
analogy between the similarity $\bar{\eta}_{\scriptscriptstyle\alpha} \equiv \bar{u}^{<-1>} \punt
\beta \punt \bar{\alpha},$ and the one characterizing the
$\alpha$-cumulant umbra $\kappa_{\scriptscriptstyle \alpha} \equiv
\chi \punt \alpha \equiv u^{<- 1>} \punt \beta \punt \alpha.$
\begin{theorem}[Parametrizations] Let $\eta_{\scriptscriptstyle\alpha}$ be the
$\alpha$-boolean cumulant umbra. For $i=1,2,\ldots$, we have
\begin{equation}
\bar{\alpha}^i \simeq  \bar{\eta}_{\scriptscriptstyle\alpha}
(\bar{\eta}_ {\scriptscriptstyle\alpha} + 2 \punt \bar{u} \punt
\beta \punt \bar{\eta}_{\scriptscriptstyle\alpha})^{i-1} \quad
\bar{\eta}_{\scriptscriptstyle\alpha}^i \simeq \bar{\alpha} (
\bar{\alpha} - 2 \punt \bar{\alpha})^{i-1}. \label{par2}
\end{equation}
\end{theorem}

\smallskip \noindent
{\bf $\alpha$-free cumulant umbra.} Let us consider a
noncommutative r.v. $X$, i.e. an element of an unital
noncommutative algebra $\cal{A}$. Suppose $\phi : \cal{A}
\rightarrow {\mathbb C}$ is an unital linear functional. The
$i$-th moment of $X$ is the complex number $m_i=\phi(X^i)$ while
its g.f. is the formal power series $M(t)= 1 + \sum_{i \geq 1} m_i
t^i$. The noncrossing (or free) cumulants of $X$ are the
coefficients $r_i$ of the ordinary power series $R(t) = 1 +
\sum_{i \geq 1} r_i t^i$ such that $M(t)=R[t M(t)]$. The umbral
theory of free cumulants has been introduced in \cite{free}. As
before, the $\bar{\alpha} \,$-free cumulant
$\mathfrak{K}_{\scriptscriptstyle \bar{\alpha}}$ has been
characterized so that $E[\mathfrak{K}_{\scriptscriptstyle
\bar{\alpha}}^i] = i! r_i.$ In particular, by using the Lagrange
inversion formula \cite{Dinardo}, the $\bar{\alpha} \,$-free
cumulant umbra is such that $(- 1 \punt
\mathfrak{K}_{\scriptscriptstyle
\bar{\alpha}})_{\scriptscriptstyle D} \equiv
\bar{\alpha}_{\scriptscriptstyle D}^{<-1>},$ where
$\alpha_{\scriptscriptstyle D}$ is the derivative umbra of
$\alpha,$ i.e. such that $\alpha_{\scriptscriptstyle D}^n \simeq n
\, \alpha^{n-1}.$ It is quite obvious to observe the difference
between the previous equivalence and the ones characterizing both
the $\alpha$-classical and the $\alpha$-boolean cumulant umbrae.
This is why the computation of free cumulants is quite difficult
compared with the classical and boolean ones. Speicher has found a
way to expressing free cumulants $\{r_n\}_{n \geq 1}$ in terms of
moments $\{m_n\}_{n \geq 1}$ (and viceversa) by using non-crossing
partitions of a set \cite{Speicher1,Speicher2}. However, the
resulting algorithm is quite difficult to implement. Bryc in
\cite{Bryc} uses a different approach. In the next section, we
propose an unifying algorithm, which relies on the following
parametrizations.
\begin{theorem}[Parametrizations] Let $\mathfrak{K}_{\scriptscriptstyle
\bar{\alpha}}$ be the $ \bar{\alpha}$-free  cumulant umbra. For
$i=1,2,\ldots$ we have
\begin{equation}
\bar{\alpha}^i \simeq
\mathfrak{K}_{\scriptscriptstyle\bar{\alpha}} (\mathfrak
{K}_{\scriptscriptstyle\bar{\alpha}} + i \punt
\mathfrak{K}_{\scriptscriptstyle\bar{\alpha}})^{i-1} \quad
\mathfrak{K}_ {\scriptscriptstyle\bar{\alpha}}^i \simeq
\bar{\alpha} ( \bar{\alpha} - i \punt \bar{\alpha})^{i-1}.
\label{par3}
\end{equation}
\end{theorem}
The umbral polynomials $x(x - n \punt \alpha)^{n-1}$ are known as umbral Abel
polynomials.
\section{The umbral algorithm}
The algorithm we propose relies on an efficient expansion of the
following umbral polynomial $\gamma (\gamma + \delta \punt
\gamma)^{i-1}$ for $i = 1,2,\ldots$ and $\delta,\gamma$ umbrae.
\begin{proposition} \label{propfond}
If $\delta, \gamma \in A$ then
\begin{eqnarray}
\gamma (\gamma + \delta \punt \gamma)^{i-1} \simeq \sum_{\mu
\vdash i} (\delta)_ {\nu_{\mu}-1} \, d_{\mu} \, \gamma_{\mu}.
\label{alg}
\end{eqnarray}

\end{proposition}
\begin{proof}
By using the binomial expansion and equivalence (\ref{(eq:11)}),
we have
\begin{equation}
\gamma (\gamma + \delta \punt \gamma)^{i-1} \simeq \sum_{s=1}^i
\left( \begin {array}{c}
i - 1 \\
s - 1
\end{array} \right) \gamma^s \sum_{\lambda \vdash \, i - s} (\delta)_{\nu_
{\lambda}} d_{\lambda} \gamma_{\lambda}.
\end{equation}
Suppose we consider the partition $\mu$ of the integer $i,$
obtained by adding the integer $s$ to the partition $\lambda$.
Then we have $\gamma^s \gamma_{\lambda} \equiv \gamma_{\mu}$ and
$\nu_{\lambda} = \nu_{\mu} - 1$. If $c_s$ denotes the multiplicity
of $s$ in $\lambda$ and $m_s$ denotes the multiplicity of $s$ in
$\mu,$ then $m_s=c_s + 1.$ Therefore, we have
$$ \left( \begin{array}{c}
i - 1 \\
s - 1
\end{array} \right) d_{\lambda} = \frac{s}{i} \,\, \frac{i!}{(1!)^{c_1} \cdots
(s!)^{c_s+1} \cdots \,\, c_1! \cdots c_s! \cdots} = s \, m_s \,\,
\frac{d_{\mu}}{i},$$ where the last equality comes by multiplying
numerator and denominator for $m_s
> 0$. Recalling that $\sum s \, m_s = i,$
equivalence (\ref{alg}) follows.
\end{proof}
In order to evaluate $\gamma (\gamma + \delta \punt \gamma)^{i-1}$ via (\ref{alg}), we need
the factorial moments of $\delta$ and the moments of $\gamma.$ Recall that,
if we just have information on moments $\delta^i$, the factorial moments can be
recovered by using the well-known change of bases:
$$(\delta)_i  \simeq \sum_{k=1}^i s(i,k) \delta^k,$$
where $\{s(i,k)\}$ are the Stirling numbers of I kind. In particular
equivalence (\ref{alg}) allows us to give any expression of cumulants (classical,
boolean, free) in terms of moments and viceversa.   \\
{\it i)} For classical cumulants in terms of moments, due to the latter
of (\ref {par1}), we  set $\delta=-1 \punt u$ and $\gamma= \alpha$. Here
we find $E[(-1 \punt u)_i] = (-1)_i = (-1)^i i!.$  \\
{\it ii)} For moments in terms of classical cumulants, due to the first of
(\ref {par1}), we set $\delta= \beta$ and $\gamma=\kappa_{\scriptscriptstyle \alpha}.$
Here we know $E[(\beta)_i] = 1.$  \\
{\it iii)} For boolean cumulants in terms of moments, due to the latter of (\ref{par2}),
we set $ \delta=-2  \punt u$ and and $\gamma = \bar{\alpha}.$ Here we find
$E[(-2  \punt u)_i] =  (-1)^i (i+1)!.$  \\
{\it iv)} For moments in terms of boolean cumulants, due to the first of
(\ref{par2}), we set $\delta = 2 \punt \bar{u} \punt \beta$ and
$\gamma = \bar{\eta}_ {\scriptscriptstyle\alpha}.$  Here we have $E[(2 \punt \bar{u} \punt \beta)_i] =
E[(2 \punt  \bar{u} \punt \beta \punt \chi)^i] = E[(2 \punt  \bar{u})^i] =  E[(\bar{u} + \bar{u}^{\prime})^i] = (i+1)!.$ \\
{\it v)} For free cumulants in terms of moments, due to the latter of (\ref{par3}), we
set $\delta =-n \punt u$ and $\gamma = \bar{\alpha}.$ Here we have $E[(-n \punt u)_i]=(-n)_i.$ \\
{\it vi)} Finally, for moments in terms of free cumulants, due to the first of
(\ref{par3}), we set $ \delta =n \punt u$ and $\gamma =
\mathfrak{K}_{\scriptscriptstyle\bar{\alpha}}$. Here we have $E[(n \punt u)_i]=(n)_i.$ \\

The umbral algorithm in {\tt MAPLE} is the following:
\begin{verbatim}
y:=combinat['partition'](i):
umbralg := proc(i,fm,y)
            i! * add(fm[j] *
                mul(((g[x[1]]/x[1]!)^x[2])/x[2]!,
                   x=convert(y[j],multiset)),
            j=1..nops(y));
 end:
\end{verbatim}
In the {\tt MAPLE} procedure, the factorial moments
$E[(\delta)_{j-1}]$ are referred by the vector {\tt fm[j]} and the
moments $E[\gamma^k]$ are referred by the vector {\tt g[k]}.

 Table 1 refers to computational times (in seconds) reached by
using the umbral algorithm and the Bryc's procedure \cite{Bryc},
both implemented in {\tt MAPLE}, release 7, when we need free
cumulants in terms of moments \footnote{The output is in the same
form of the one given by Bryc's procedure.}.
\begin{table}[!htp]
\begin{tabular}{l c c} \hline
$i$ & {\tt MAPLE} (umbral) & {\tt MAPLE} (Bryc)\\
\hline
15 & 0.015 & 0.016\\
18 & 0.031 & 0.062\\
21 & 0.078 & 0.141\\
24 & 0.172 & 0.266\\
27 & 0.375 & 0.703\\
\hline
\end{tabular}
\caption{Comparisons of computational times needed to compute free
cumulants in terms of moments. Tasks performed on Intel (R)
Pentium (R),CPU 3.00 GHz, 512 MB RAM.}

\end{table}
\\
{\bf Moments of Wigner semicircle distribution.} In free
probability, the Wigner semicircle distribution is analogous to
the Gaussian r.v. in the classical probability. Indeed, free
cumulants of degree higher than 2 of the Wigner semicircle r.v.
are zero. The first column in Table 2 shows moments of the Wigner
semicircle r.v. $X$, computed by the umbral algorithm. They are
compared with Catalan numbers $C_i$ (second column), since it is
well-known that $E[X^{2i}]=C_i$ and $E[X^{2i+1}]=0$. By using
equivalence (\ref{(eq:14)}) it is straightforward to prove that
the umbra corresponding to the Wigner semicircle distribution is
$\bar{\varsigma} \punt \beta \punt \bar{\delta},$ where
$\varsigma$ is the umbra whose moments are the Catalan numbers and
$\delta$ is an umbra having only second moment equal to $1,$ the
others being zero. In the next section,
we will use again this umbra in describing the umbral syntax of a Gaussian r.v.\\
{\bf Moments of Marchenko-Pastur distribution.} In free
probability, the Marchenko-Pastur distribution is analogous to the
Poisson r.v. in the classical probability.  Indeed, the free
cumulants are all equal to a parameter $\lambda$. The last column
in Table 2 shows moments of the Marchenko-Pastur distribution
computed by the umbral algorithm.
\begin{table}[!htp]
\begin{tabular}{l c c l}
\hline
$i$ & Wigner & Catalan & Marchenko-Pastur r.v. \\
  &   r.v. & numbers &  \\
\hline
$2$ & $1$ & $2$ & $\lambda^{2}+\lambda$\\
$3$ & $0$ & $5$ & $\lambda^{3}+3\lambda^{2}+\lambda$\\
$4$ & $2$ & $14$ & $\lambda^{4}+6\lambda^{3}+6\lambda^{2}+\lambda$\\
$6$ & $5$ & $132$ & $\lambda^{6}+15\lambda^{5}+50\lambda^{4}+50\lambda^{3}+15\lambda^{2}
+\lambda$\\
$8$ & $14$ & $1430$ & $\lambda^{8}+28\lambda^{7}+196\lambda^{6}+490\lambda^{5}+490
\lambda^{4}+196\lambda^{3}+28\lambda^{2}+\lambda$\\
\hline
\end{tabular}
\caption{Moments of some special free distributions.}
\end{table}
\section{Computing cumulants of some known laws through the umbral algorithm}
An umbra looks like the framework of a r.v. with no reference to
any probability space, just looking at moments. The way to
recognize the umbra corresponding to a r.v. is to characterize the
sequence of moments $\{a_n\}$. When the sequence exists, this can
be done by comparing the moment generating function (m.f.g.) of a
r.v. with the so-called generating function (g.f.) of an umbra.
The g.f. of an umbra has been defined \cite{Dinardoeurop} as the
following formal power series
$$f(\alpha,t) = 1 + \sum_{n \geq 1} a_n \frac{t^n}{n!}.$$
In the classical umbral calculus, the convergence of the formal
power series $f(\alpha,t)$ is not relevant. This means that we can
define the umbra whose moments are the same as the moments of a
lognormal r.v., even if this r.v. does not admit a m.g.f.  So the
umbral algorithm allows us to compute classical \footnote{Recall
that, due to their properties of additivity and invariance under
translation, the cumulants are not necessarily connected with the
moments of probability distributions. We can define cumulants of
any r.v. disregarding the question whether its m.g.f. converges
\cite{Feller}.}, boolean or free cumulants of any r.v. having
sequence of moments $\{m_n\}$. Once the umbra corresponding to the
r.v. has been characterized, we choose $\delta=-1 \punt u, -2
\punt u, -n \punt u$ in equivalence (\ref{alg}), depending on
whether we need classical, boolean or free cumulants.

In the following, we take up again some of the examples given in \cite{Bryc},
showing how they can be recovered through the umbral algorithm by a suitable characterization
of the involved umbrae. \\
{\bf Poisson r.v.} A Poisson r.v. of parameter $\lambda$ is
umbrally represented by the umbra $\lambda \punt \beta$, because
$f(\lambda \punt \beta, t) = \exp[\lambda(\exp(t)-1)]$, which is
the m.g.f. of a Poisson r.v. The moments are \cite{Dinardo}
$a_{i}=E[(\lambda.\beta)^{i}]= \sum _{k=1}^{i}S(i,k)\lambda^{k},$
where $S(i,k)$ are the Stirling numbers of second type. So
cumulants of a Poisson r.v. can be computed via the umbral
algorithm, taking as input the sequence of moments
$E[(\lambda.\beta)^{i}]$. If the input is the sequence of
factorial moments $\{(-1)^i i!\},$ we get classical cumulants; if
the input is the sequence of factorial moments $\{(-1)^i
(i+1)!\},$ we get boolean cumulants;  if the input is the sequence
of factorial moments $\{(-n)_i\}$ we get free cumulants (Tables 3 and 4 in \cite{Bryc}). \\
{\bf Compound Poisson r.v.} A compound Poisson r.v. $S_N = X_1 +
\cdots + X_N,$ where $\{X_i\}$ are i.i.d. r.v.'s and $N$ is a
Poisson r.v. of parameter $\lambda,$ is umbrally represented by
the umbra $\lambda \punt \beta \punt \alpha$, because $f(\lambda
\punt \beta \punt \alpha, t) = \exp\{\lambda[f(\alpha,t)-1]\}$,
which is the m.g.f. of a compound Poisson r.v. The formal power
series $f(\alpha,t)$ corresponds to the m.g.f. of $X_i$. Due to
equivalence (\ref{(eq:14)}),  the moments are $a_{i}=E[(\lambda
\punt \beta \punt \alpha)^{i}]= \sum _{\lambda \vdash i} \lambda^k
d_{\lambda} \alpha_{\lambda}.$   \\
{\bf Exponential r.v.} An exponential r.v. is umbrally represented
by the umbra $\bar{u} \over \lambda$ because $f(\bar{u} / \lambda,
t) = (1- {t \over \lambda})^{-1},$ which is the m.g.f. of an
exponential r.v. with parameter $\lambda> 0.$ So its moments are
$a_i = E[({\bar{u} \over \lambda})^n] = {n! \over \lambda}.$ In
order to obtain the second column of Table 3 in \cite{Bryc},
choose in (\ref{alg}) $\lambda=1$ and $\{(-n)_i\}$ as
factorial moments. \\
{\bf Uniform r.v.} An uniform r.v. on the interval $[a,b]$ has
m.g.f.
$$M(t)=\frac{e^{t \,b} - e^{t \,a}}{t(b-a)} = e^{t \,a} \left[
\frac{e^s-1}{s} \right]_{s=t(b-a)}.$$ The umbra, with g.f.
${e^s-1} \over s$, is the inverse of the Bernoulli umbra
\cite{SIAM}, i.e. $-1 \punt \iota.$ So an uniform r.v. on the
interval $[a,b]$ is umbrally represented by the umbra $a+(b-a)(-1
\punt \iota)$. Recalling that $E[(-1 \punt \iota)^i] = {1 \over
{i+1}}$, we get
$$a_{i}=E\bigl\{[a+(b-a)(-1.\iota)]^{i}\bigr\}=\sum_{j=0}^{i}{i\choose
j}a^{i-j}\frac{(b-a)^{j}}{j+1}.$$ In order to obtain the last
column of Table 3 in \cite{Bryc}, choose in (\ref{alg}) $a=-1,
b=1$ and $\{(-n)_i\}$ as factorial moments. \\
{\bf Bernoulli r.v.} A Bernoulli r.v. of parameter $p \in (0,1)$
has m.g.f. $M(t)= q + p e^t = 1 + p(e^t - 1).$ The umbra with g.f.
$M(t)$ is $\chi \punt p \punt \beta$ (see \cite{Dinardoeurop} for
more details), whose moments are $a_i = p$. \\
{\bf Binomial r.v.} A binomial r.v. of parameters $n,$ a positive
integer, and $p \in (0,1)$ is a sum of $n$ i.i.d. Bernoulli
r.v.'s. Similarly a binomial r.v. is umbrally represented by $n
\punt (\chi \punt p \punt \beta)$. Due to equivalence
(\ref{(eq:11)}), we have $a_i = E\{[n.(\chi \punt p \punt
\beta)]^i\} = \sum_{\lambda\vdash i}(n)_{\nu_{\lambda}}
d_{\lambda} p^{\nu_{\lambda}}$. In order to obtain the second
column of Table 4 in \cite{Bryc}, choose $\{(-n)_i\}$ as
factorial moments  in (\ref{alg}). \\
{\bf Gaussian r.v.} A gaussian r.v. with real parameter $\mu$ and
$\sigma > 0$  has m.g.f. $M(t)=\exp \left(\mu t + \sigma^2 {t^2
\over 2} \right).$ This power series is the g.f. of the umbra $\mu
+ \beta \punt (\sigma \delta)$, where $\delta$ is an umbra such
that $E[\delta^2]=1$ whereas $E[\delta  ^i]=0$ for positive
integers $i \ne 2.$ The following proposition gives the expression
of the $n$-th moment of the umbra representing the Gaussian r.v.
\begin{proposition}
For n=1,2,\ldots we have
\begin{equation}
a_n= E\big[(\mu+\beta.(\sigma\delta))^n\big]=
\sum_{k=0}^{\lfloor {n / 2} \rfloor}\bigg({\sigma^2 \over 2}\bigg)^k \frac{(n)_{2k}}{k!}\mu^{n-2k}.
\label{(mgau)}
\end{equation}
\end{proposition}
\begin{proof}
By using the binomial expansion, we have
\begin{equation}
\big(\mu+\beta.(\sigma\delta)\big)^{n}\simeq
\sum_{j=0}^{n}{n\choose
j}\mu^{n-j}\big[\beta.(\sigma\delta)\big]^{j} \simeq
\sum_{j=0}^{n}{n\choose
j}\mu^{n-j} \sum_{\lambda\vdash
j}d_{\lambda} (\sigma\delta)_{\lambda}.
\label{(be)}
\end{equation}
Suppose $\lambda=(1^{r_1}, 2^{r_2}, \ldots),$ then
$(\sigma\delta)_{\lambda} \equiv (\sigma\delta^{\prime}
)^{.r_{1}}[\sigma^2 (\delta^{\prime \prime})^2]^{.r_{2}} \cdots
\equiv \sigma^{j} \delta_{\lambda}$. On the other hand, due to the
definition of the umbra $\delta$, we have $E[\delta_{\lambda}]\neq
0$ iff the partition $\lambda$ is of type $(2^{r_2})$. Thus, if
$j$ is odd, there does not exist any partition of this type and so
$\big[\beta.(\sigma \delta)\big]^{j}\simeq 0$. Instead, if $j$ is
even, say $j=2k$, there exists a unique partition of type
$(2^{r_2})$ which corresponds to $r_2 = k$. For this partition, we
have $d_{\lambda} = {{(2k)!} \over {(2!)^{k} \, k!}}$ so that
$\big[\beta.( \sigma \delta)\big]^{2k} \simeq {{(2k)! \,
\sigma^{2k}} \over {(2!)^{k} \, k!}}$. Replacing this last
equivalence in (\ref{(be)}), we have
$$\big(\mu+\beta.(\sigma\delta)\big)^{n} \simeq
\sum_{k=0}^{\lfloor {n / 2} \rfloor}\frac{n!}{(2k!)(n-2k)!} \mu^{n-2k}\frac{(2k!)\sigma^{2k}}
{(2!)^{k}k!}$$
by which the result follows.
\end{proof}
We have proved by umbral tools that moments
of a Gaussian r.v. can be expressed by using a particular
sequence of orthogonal polynomials, the \emph{Hermite polynomials}:
\begin{equation*}
H_{n}^{(\nu)}(x)\simeq
\sum_{k=0}^{\lfloor {n / 2} \rfloor}\bigg(\frac{-\nu}{2}\bigg)^k \frac{(n)_{2k}}{k!}x^{n-2k},
\end{equation*}
with $x=\mu$ and $\nu=-\sigma^{2}$. All the properties of Hermite
polynomials can be easily recovered by using the Gaussian umbra.
We skip the details. As before, the first column of Table 3 in
\cite{Bryc} can be recovered from the umbral algorithm by using
(\ref{(mgau)}) with $\mu=0$ and $\sigma=1$ and choosing
$\{(-n)_i\}$  as factorial moments in (\ref{alg}).
\section{Acknowledgments}
The authors thank Giuseppe Guarino for his contribution in the implementation
of the umbral algorithm.

\end{document}